\documentclass[conference,letterpaper]{ieeeconf}

\usepackage{amsmath,amssymb,amsfonts,amsbsy,mathrsfs}
\usepackage{blindtext}
\usepackage{graphicx}
\usepackage{color}
\usepackage{float}
\usepackage{multirow}
\usepackage{cite}
\usepackage{setspace}
\renewcommand{\baselinestretch}{0.92}

\newcommand{\R}{\mathbb R}
\newcommand{\N}{\mathbb N}

\newtheorem{assumption}{Assumption}
\newtheorem{remark}{Remark}
\newtheorem{theorem}{Theorem}

\title{A Sliding Mode Observer Approach for Attack Detection and Estimation in Autonomous Vehicle Platoons using Event Triggered Communication}
\author{%
Twan Keijzer \\
\textit{Delft Centre for Systems and Control}\\
\textit{Delft University of Technology}\\
{\tt\small t.keijzer@tudelft.nl}%
\and
Riccardo M.G. Ferrari \\
\textit{Delft Centre for Systems and Control}\\
\textit{Delft University of Technology}\\
{\tt\small r.ferrari@tudelft.nl}%
}
\begin{document}
\maketitle
\begin{abstract}
Platoons of autonomous vehicles are being investigated as a way to increase road capacity and fuel efficiency. Cooperative Adaptive Cruise Control (CACC) is one approach to controlling platoons longitudinal dynamics, which requires wireless communication between vehicles. In the present paper we use a sliding mode observer to detect and estimate cyber-attacks threatening such wireless communication. In particular we prove stability of the observer and robustness of the detection threshold in the case of event-triggered communication, following a realistic Vehicle-to-Vehicle network protocol.
\end{abstract}
\section{Introduction}\label{sec:introduction}
\noindent Autonomous vehicle {\em platoons} and {\em Cooperative Adaptive Cruise Control} (CACC) are topics that received significant attention by researchers in recent years \cite{Wu2016-lz, Hu2018-ua, Chen2015-mo, Somarakis2018-eg,naus2010cooperative,Ploeg2011-nc}. CACC is a longitudinal cooperative control technique that allows platoons, or strings, of autonomous vehicles to coordinate themselves. The goal is to have vehicles in the platoon travelling closer together than human drivers, or non-cooperative control approaches like {\em Adaptive Cruise Control}, can. Benefits of this lower inter-vehicle spacing include better fuel efficiency and road utilization.
Vehicles in a CACC platoon measure relative position and velocity of the preceding vehicle, and also communicate (see figure~\ref{cacc}) in order to attain {\em string stability}, which is an important property resulting in dampening of velocity changes down the platoon \cite{Ploeg2011-nc}.
\vspace{-0.1cm}
\begin{figure}[H]
	\centering
	\includegraphics[width=\columnwidth,trim={0cm, 1.1cm, 0cm, 1.1cm},clip]{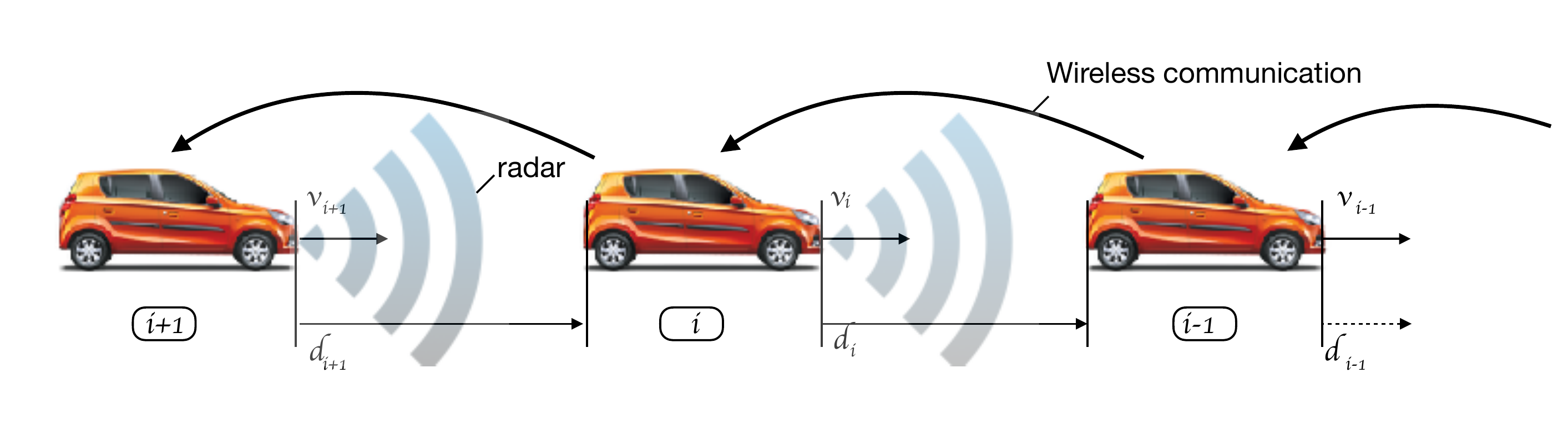}
	\caption{CACC equipped string of vehicles. The V2V communication is implemented wirelessly, and is subjected to a class of cyber attacks.}
	\label{cacc}
\end{figure}%
\vspace{-0.2cm}
The reliance of CACC platoons on inter-vehicle wireless communications, be it {\em periodic} or {\em event-triggered } \cite{Linsenmayer2015-ba,Dolk2017-wa,Proskurnikov2018-at}, may expose them to the same kind of threats as other networked control systems or {\em Cyber-Physical Systems} (CPS), such as {\em Denial of Service} (DoS), routing, replay and stealthy data injection attacks (see \cite{Cardenas2008-pp,Teixeira2015-nw}). Indeed, vulnerabilities of {\em Vehicle-to-Vehicle} (V2V) networks to cyber attacks have been investigated in \cite{ Studnia2013-dk,Miller2014-op,Amoozadeh2015-qp,Ploeg2017-zj}. While CACC can provide limited robustness to network induced effects such as random packet losses (see \cite{Lei2011-qq, Ploeg2013-iw}), the case of a malicious attacker targeting the (V2V) network  should be addressed by dedicated detection and fault-tolerant control methods.

While the case of faults in autonomous vehicles formations was addressed in \cite{Meskin2009-wu} and \cite{Quan2018-wd} with an observer-based approach, few works dealt with cyber-attacks. \cite{ Abdollahi_Biron2018-vu} considered the problem of designing a model based observer for detecting DoS attacks, which were characterised as an equivalent time delay in the communication network.\\
\indent In this paper we are going to extend some preliminary results presented by the authors in \cite{JAHANSHAHI2018212}, where a {\em Sliding-Mode Observer} (SMO) was introduced for estimating false data injection attacks. The contribution of the paper is twofold: we prove the stability of the SMO under event-triggered communication and less restrictive assumptions on measurement uncertainties, and we introduce robust adaptive attack detection thresholds for such a scenario. In particular, we will assume the vehicle platoon is using a realistic event-triggered communication protocol based on the current ETSI-ITS G5 V2V communication standard \cite{Lyamin2018-ht,European_Telecommunications_Standards_Institute2014-iz}.\\
\indent The use of sliding mode observers for fault detection was pioneered by \cite{Edwards2000-vs} and developed further by \cite{Floquet2004-vg,Yan2007-gl}, amongst others. By monitoring the so-called {\em equivalent output injection} (EOI), this method allows to estimate actuator and sensor faults or, as in \cite{JAHANSHAHI2018212} and the present case, a false data injection attack. Previous results considered continuous communication, and did not derive an adaptive detection threshold guaranteed to be robust against uncertainties or communication-induced effects. The literature on fault detection for event-triggered systems, instead, includes works such as \cite{Liu2013-ax,Sid2014-ko,Davoodi2017-ru}, which are concerned with the simultaneous design of the triggering condition and the fault detector, 
while \cite{Boem2017-yr} addressed the case of asynchronous communication and packet loss for fault detection of networked control systems.\\
\indent While several works considered the case of event-triggered sliding mode control, such as \cite{Cucuzzella2016-zx,Incremona2016-ci,Wu2017-gr,Behera2017-bp}, the present approach would be, to the best of the authors knowledge, the first contribution considering sliding mode observers for fault, or cyber-attack detection and estimation in systems where event--triggered communication is present.\\
\indent The remainder of the paper is organized as follows. Section~\ref{sec:problem_formulation} introduces event-triggered CACC for a vehicle platoon and describes the attack and its effect on the platoon. Section~\ref{sec:SMO} presents the sliding mode observer and characterizes its stability, and section~\ref{sec:thresholds} presents the attack detection threshold and provides theoretical results on its robustness. Section~\ref{sec:attack_est} provides preliminary results on attack estimation. In sections~\ref{sec:results} and \ref{sec:conclusion}, respectively, the simulation results, and conclusion and future work are presented.
\subsection{Notation}
\noindent Throughout the paper, a notation such as $x_i$ will denote a variable $x$ pertaining to the $i$--th vehicle, while $x_{i,(j)}$ will denote the $j$--th component of the vector $x_i$.
\section{Problem Formulation}\label{sec:problem_formulation}
\subsection{Error Dynamics of a Platoon using CACC}
\noindent In the present paper we will use the CACC formulation in \cite{Ploeg2011-nc} and its extension to event triggered communication introduced in \cite{Dolk2017-wa}, while the event-triggering condition will follow \cite{Lyamin2018-ht,European_Telecommunications_Standards_Institute2014-iz}.
We will consider a string of $m\in \N$ homogeneous vehicles (see Figure~{\ref{cacc}}), each modeled as
\begin{equation} \label{eq:vehicle_dynamics}
	\begin{bmatrix} \dot{p}_i(t) \\
				\dot{v}_i(t) \\
				\dot{a}_i(t)
	\end{bmatrix} =
	\begin{bmatrix} v_i(t) \\
				a_i(t) \\
				\frac{1}{\tau}(u_i(t)-a_i(t))
	\end{bmatrix},
\end{equation}
where $p_i(t)$, $v_i(t)$, $a_i(t)$ and $u_i(t) \in \R$ are the position, velocity, acceleration and the input of the $i$-th vehicle, respectively; furthermore, $\tau$ represents the engine's dynamics. Each vehicle is assumed to measure its own local output $y_{i} \triangleq [p_i \; v_i \; a_i]^\top + \xi_i$ and, with its front radar, the relative output $y_{i,i-1} \triangleq [d_i \; \Delta v_i]^\top + \eta_i$,
where  $d_i(t) \triangleq \left( p_{i-1}(t) \!- \! p_i(t) \!- \! L  \right)$ is the inter-vehicle distance, $L$ is the length of each vehicle, $\Delta v_i \triangleq v_{i-1} - v_i$ is the relative velocity and $\xi_i$ and $\eta_i$ are the measurement uncertainties affecting the vehicle sensors.
\begin{assumption}\label{ass:measurement_uncertainties}
For each $i$--th vehicle, the measurement uncertainties $\xi_i$ and $\eta_i$ are unknown but they are upper bounded by known quantities $\bar \xi_{i}$ and $\bar \eta_{i}$, i.e. $\vert \xi_{i,(j)}(t) \vert \leq \bar \xi_{i,(j)}(t)$ and $\vert \eta_{i,(j)}(t) \vert \leq \bar \eta_{i,(j)}(t)$ for all $j$, and all $t$.
\end{assumption}

\begin{equation}\label{eq:spacing_policy}
d_{r,i}(t)=r_i+hv_i(t) \,,
\end{equation}
while making the relative velocity $\Delta v_{i}$ tend to zero in steady state.
in eq. \eqref{eq:spacing_policy} $r_i$ and $h$ are the desired distance at stand still, and the time headway between the vehicles respectively. \cite{Ploeg2011-nc}

Let us introduce the position error $e_i(t) \triangleq d _i(t)-d_{r,i}(t)$ and its time derivative $\dot{e}_i(t) = \Delta v_i - h a_i(t)$. 
In \cite{Ploeg2011-nc}, a CACC control law is initially proposed in ideal conditions, as the solution to the following equation

\begin{equation} \label{eq:ideal_local_CACC_control}
		\dot{u}_i(t) = \frac{1}{h} \left[ -u_i(t) + \left( k_p e_{i}(t) +k_d \dot e_{i}(t)\right) + u_{i-1}(t) \right].
\end{equation}
As can be seen from Eq.~\eqref{eq:ideal_local_CACC_control}, the local control law depends on measured quantities, such as the relative position and velocity, which will be corrupted by noise. Furthermore, the control law depends on the intended acceleration of the preceding vehicle, $u_{i-1}(t)$, which shall be received through a wireless V2V communication network.

In this paper the presence of measurement uncertainties and non-ideal communication are explicitly incorporated in the control law giving
\begin{equation} \label{eq:local_CACC_control}
		\dot{u}_i(t) = \frac{1}{h} \left[ - u_i(t) + \left( k_p \hat e_{i}(t) +k_d \hat {\dot e}_{i}(t)\right) + \tilde u_{i-1}(t) \right],
\end{equation}
where $\hat e_{i} \triangleq e_{i} + \eta_{i,(1)} - h \xi_{i,(2)}$, $\hat {\dot e}_{i} \triangleq \dot e_{i} + \eta_{i,(2)} - h \xi_{i,(3)}$, and $\tilde u_{i-1}(t)= u_{i-1}(t) + \Delta u_{i-1}(t)$ is the last received value of $u_{i-1}(t)$. $\Delta u_{i-1}$ will be further defined in subsection \ref{ssec:attackeffect}.

By following similar steps as in \cite{Ploeg2011-nc} and \cite{JAHANSHAHI2018212}, we can write the  $i$--th vehicle error dynamics, under control law~\eqref{eq:local_CACC_control},  as
\begin{equation}\label{eq:error_dynamics}
\mathcal{E}_i:\left\{
    \begin{aligned}
        \dot x_{e_i}(t) & = A_e x_{e_i}(t) + B_e \zeta_i(t)\\
        y_{e_i}(t)  & = C_e x_{e_i}(t) + D_e \zeta_i(t)
    \end{aligned} 
\right.\,, \\
\end{equation}
where $C_e = D_e$ and the following quantities were introduced
\begin{equation}\begin{gathered}\label{eq:error_dynamics_definitions}
\hspace{-0.15cm}A_e \triangleq
\begin{bmatrix} 
0  & 1& 0\\
0 & 0 & 1\\
-\frac{k_p}{\tau}&-\frac{k_d}{\tau}  &  -\frac{1}{\tau}
\end{bmatrix}, \;
B_e \triangleq
\begin{bmatrix} 
0  & 0& 0\\
0 & 0 & 0\\
-\frac{k_p}{\tau}&-\frac{k_d}{\tau}  &  -\frac{1}{\tau}
\end{bmatrix} \\
\hspace{-0.5cm}C_e \triangleq \begin{bmatrix} 
1 &\!0\\
0 &\!1\\
0 &\!0
\end{bmatrix}^\top \!\!\!, \;
x_{e_i} \!\triangleq\!
\begin{bmatrix} 
e_i(t)\\
\dot{e}_i(t)\\
\ddot{e}_i(t)
\end{bmatrix}, \;
\zeta_i \!\triangleq\! \begin{bmatrix}
\eta_{i,(1)} - h \xi_{i,(2)}\\
\eta_{i,(2)} - h \xi_{i,(3)}\\
\Delta u_{i-1}(t)
\end{bmatrix}
\end{gathered}\end{equation}
The stability and performance of the error dynamics $\mathcal E_i$ and the string-stability of the platoon have been analysed in \cite{Ploeg2011-nc} and \cite{Dolk2017-wa}. As the present paper is concerned with the design of a cyber-attack detection and estimation scheme, and not the event-triggered CACC control scheme itself, for well-posedness we will require the following
\begin{assumption}\label{ass:stability}
Control law $u_i$ (Eq.~\eqref{eq:local_CACC_control}) and triggering condition $\sigma$ (Eq.~\eqref{eq:triggering_condition}) are chosen such that, without cyber-attacks and when Assumption~\ref{ass:measurement_uncertainties} holds, $\mathcal E_i$ is stable for each vehicle $i$ and string stability of the platoon is guaranteed.
\end{assumption}
\subsection{Attack and communication-induced effects}
\label{ssec:attackeffect}
\noindent In this paper, following 
\cite{Dolk2017-wa,Lyamin2018-ht,European_Telecommunications_Standards_Institute2014-iz}, the transmission of $u_{i-1}$ is assumed to be event triggered. Furthermore a {\em man-in-the-middle attack} on the transmitted $u_{i-1}$ is considered. We are not interested here in the actual implementation of the attack, for this, one can refer to \cite{Studnia2013-dk,Miller2014-op,Amoozadeh2015-qp,Ploeg2017-zj}. For the observer, the effects of communication, $\Delta u_{i-1,C}(t)$, and the attack, $\phi_i(t)$, will be combined in $\Delta u_{i-1}(t) = \Delta u_{i-1,C}(t) + \phi_i(t)$. 

The event-triggered communication causes a variable delay in the signal received by car $i$, defined as
\begin{equation}
\tau_0 = 0, \tau_{l+1} \triangleq \inf \left\{ t\geq \tau_l : \, \sigma = 1 \right\},
\end{equation}
where $\tau_l$ is the last transmission time, and $\sigma$ is a triggering condition based on the local measurements, $y_{i-1}$, in car $i-1$:
\begin{equation}
\begin{aligned}
\sigma &\triangleq (t-\tau_l\geq T_H \lor (t-\tau_l>T_L \land\\
&\hspace{-0.2cm}\exists j=\{1,2\} : \,\vert y_{i-1,(j)}(\tau_l) - y_{i-1,(j)}(t)\vert \geq \Delta y_{L,(j)})).
\end{aligned}
\label{eq:triggering_condition}
\end{equation}
Here $T_L$, $T_H$ and $\Delta y_L \in \R^2$ are user-designed parameters that define, respectively, the minimum and maximum inter-triggering times, and the threshold for communication.

In summary, communication is triggered on changes in local measurements of car $i-1$ since the last communication. This is combined with a minimum and maximum inter-triggering time. The error introduced by the event-triggered communication is denoted by $\Delta u_{i-1,C}(t)$.
\section{Sliding Mode observer}\label{sec:SMO}
\noindent In this section a {\em Sliding Mode Observer} (SMO) for the dynamics $\mathcal E_i$ in eq. \eqref{eq:error_dynamics} is presented. To this end, first the change of variables $z_{1,i} =
\begin{bmatrix}
x_{e_i,(1)}\\
x_{e_i,(2)}
\end{bmatrix},~
\zeta_{1,i} = \begin{bmatrix}
\zeta_{i,(1)}\\
\zeta_{i,(2)}
\end{bmatrix},~z_{2,i} = x_{e_i,(3)},~ b = -\frac{1}{\tau}$ is performed in order to separate the measured and unknown states, giving:
\begin{equation}\label{eq:reduced_linsys}
\left[\begin{matrix}
\dot{z}_{1,i} \\
\dot{z}_{2,i}
\end{matrix}\right]
=
\left[\begin{matrix}
A_{11}&A_{12}\\
A_{21}&A_{22}
\end{matrix}\right]
\left[\begin{matrix}
{z}_{1,i}\\
{z}_{2,i}
\end{matrix}\right]
+
\left[\begin{matrix}
0\\
A_{21}\zeta_{1,i} + b \Delta u_{i-1}
\end{matrix}\right],
\end{equation}
\begin{equation}
    y_{e_i} = z_{1,i} + \zeta_{1,i}.
    \label{eq:gamma}
\end{equation}
An observer design is presented, in eqs. \eqref{eq:observer} and \eqref{eq:pseudocontrol}, to make the states slide along $\epsilon_{y,i}(t) = 0$ even in the presence of noise-, communication- and attack-induced effects.
\begin{equation}
\left[\begin{matrix}
\hat{\dot{z}}_{1,i} \\
\hat{\dot{z}}_{2,i}
\end{matrix}\right]
=
\left[\begin{matrix}
A_{11}&A_{12}\\
A_{21}&A_{22}
\end{matrix}\right]
\left[\begin{matrix}
\hat{z}_{1,i}\\
\hat{z}_{2,i}
\end{matrix}\right]
-
\left[\begin{matrix}
\nu_i\\
0
\end{matrix}\right]
\label{eq:observer}
\end{equation}
\begin{equation}\label{eq:pseudocontrol}
\nu_i(t)=(A_{11}+P)\epsilon_{y,i}(t)+M_i\text{sgn}(\epsilon_{y,i}(t))
\end{equation}
Here $M_i$ is a positive constant, and $P \in \R^{2\times2}$ is a positive definite matrix. Both are chosen to they verify the hypothesis of Theorem \ref{theorem:bound_eps1}, to guarantee the SMO stability. The observer error dynamics can be written as in eqs.  \eqref{eq:obs_error}, \eqref{eq:obs_error_dyn}.
\begin{equation}\label{eq:obs_error}
\begin{matrix}
\epsilon_{1,i}(t)=\hat{z}_{1,i}(t)-z_{1,i}(t)\\
\epsilon_{2,i}(t)=\hat{z}_{2,i}(t)-z_{2,i}(t)\\
\epsilon_{y,i}(t)=\hat{z}_{1,i}(t)-(z_{1,i}(t)+\zeta_{1,i}) = \epsilon_{1,i}(t)-\zeta_{1,i}
\end{matrix}
\end{equation}
\begin{equation}\label{eq:obs_error_dyn}
\hspace{-0.05cm}\dot{\epsilon}_i(t) = \left[\begin{matrix}
A_{11}&A_{12}\\
A_{21}&A_{22}
\end{matrix}\right]\epsilon_i(t) -
\left[\begin{matrix}
\nu_i(t)\\
A_{21}\zeta_{1,i}(t) + b\Delta u_{i-1}(t)
\end{matrix}\right]
\end{equation}
\begin{theorem}
\label{theorem:bound_eps1}
$\epsilon_{1,i}(t)$, under the observer dynamics in \eqref{eq:obs_error_dyn}, can be bounded by $\bar{\epsilon}_{1}=\bar{\zeta}$ if $M_i>\left|A_{12} \bar{\epsilon}_{2,i} \right| +\left|A_{11}\bar{\zeta} \right|$.
\end{theorem}
\begin{proof}
This proof will only consider the upper bound of $\epsilon_{1,i}(t)$, the lower bound can be proved in a similar manner. It will be proven that if $\epsilon_{1,i}>\bar{\zeta}$, then $\dot{\epsilon}_{1,i}<0$. This is sufficient to prove \hspace{-0.07cm}$\bar{\zeta}\geq \epsilon_{1,i}~\forall~ t$. First note that $\epsilon_{1,i}>\bar{\zeta}$ implies $\epsilon_{y,i}>0$, so the first row of eq. \eqref{eq:obs_error_dyn} can be rewritten to
\begin{equation}\label{eq:obs_err_dyn_12}
\dot{\epsilon}_{1,i} = P(\zeta_{1,i} - \epsilon_{1,i}) + A_{11} \zeta_{1,i} + A_{12} \epsilon_{2,i} - M_i
\end{equation}
Substituting the condition on $M_i$ gives
\begin{equation}
\begin{aligned}
\dot{\epsilon}_{1,i} <& P(\zeta_{1,i} - \epsilon_{1,i})+ (A_{11} \zeta_{1,i} - |A_{11} \bar{\zeta}|)\\
 & + (A_{12} \epsilon_{2,i}-|A_{12} \bar{\epsilon}_{2,i}|) \leq 0
\end{aligned}
\end{equation}
$\bar{\zeta}$, $\bar{\epsilon}_{2,i}$ and other bounds are proven in the appendix.
\end{proof}

In this paper, as in \cite{Edwards2000-vs} and subsequent works on SMO-based fault estimation, the EOI, derived from $\nu_i$, will be used for estimating attacks \cite{Edwards2000-vs}. The EOI used here will be obtained from the filter in eq. \eqref{eq:nu_filter} \cite{Utkin1992-td}.
\begin{equation}\label{eq:nu_filter}
\nu_{i,fil} = \frac{K}{s+K} \nu_i \, ,
\end{equation}
where $K>0$ is a design constant and $s$ is the Laplace domain complex variable.
\section{Attack Detection Thresholds}\label{sec:thresholds}%
\noindent As a novel contribution, we are introducing two pairs of robust attack detection thresholds on $\nu_{i,fil}$, which are guaranteed against false alarms, even in the presence of measurement uncertainties and event-triggered communication. Each pair will comprise an upper and a lower bound on the values of $\nu_{i,fil}$ in non-attacked conditions. The two pairs are termed {\em One-Switch-Ahead} (OSA) and {\em Multiple-Switches-Ahead} (MSA) thresholds, for reasons that will be apparent in next sections. For the sake of clarity, in Subsections~\ref{sub:one_switch_ahead_thresh} and \ref{sub:multiple_switch_ahead_thresh} we will assume there is no event-triggered communication, i.e. $\Delta u_{i-1,C}(t) = 0$. The effects of its presence on the thresholds will be illustrated in Subsection~\ref{sub:event_triggered_thresh}.

For the sake of notation, we will assume that the SMO is initialized at time $t_0$, and that $\text{sgn}(\epsilon_{y,i}(t_0)) = 1$. This means that between $t_0$ and the next switch at $t_1$, and all following odd intervals $[t_{2k}~t_{2k+1}]$, with $k\in \mathbb N$, the discontinuous term $\nu_i$ and $\epsilon_{y,i}(t)$ are positive, $\nu_{i,fil}$ will be increasing, and $\epsilon_{1,i}(t)$ will be decreasing. This is also shown in Figure~\ref{fig:thresholds}. Furthermore $\nu_{i,fil}$ will be initialised at $\nu_{i,fil}(t_0) = 0$ and we will denote a threshold value calculated at $t_k$ by $\bar{\nu}_{i,fil}(t_{k})$. For brevity, we will derive only the upper bound of each threshold, which is of interest in the odd time intervals, as the lower bounds and the behaviour during even time intervals can be obtained via similar reasoning.
\subsection{One-Switch-Ahead (OSA) Threshold}\label{sub:one_switch_ahead_thresh}
\noindent Let us consider the behaviour of $\nu_{i,fil}$ during the odd interval, $[t_{2k} \; t_{2k+1}]$ (see Figure~\ref{fig:thresholds}a). By introducing, in eq. \eqref{eq:max_pseudo_control}, the upper bound $\bar \nu$ on $\nu_i$, the time domain solution to \eqref{eq:nu_filter} can be upper bounded during the interval as in eq. \eqref{eq:nu_filter_resp}.
\begin{equation}\label{eq:max_pseudo_control}
\bar{\nu} = \left|(A_{11}+P)(\bar{\epsilon}_{1}+\bar{\zeta})\right| + M_i
\end{equation}
\begin{equation}\label{eq:nu_filter_resp}
\nu_{i,fil}(t) \leq e^{-K(t-t_{2k})} \nu_{i,fil}(t_{2k}) + (1-e^{-K(t-t_{2k})})\bar{\nu}
\end{equation}
\begin{remark}
	\hspace{-0.17cm} The right-hand side of eq. \eqref{eq:nu_filter_resp} is an upper bound for $\nu_{i,fil}(t)$. However, it can be easily proved that the inequality in eq. \eqref{eq:nu_filter_resp} will also hold in case of an attack. Therefore, it is not a valid threshold for attack detection.
\end{remark} 

Next, in eq. \eqref{eq:nu_filter_resp}, the hypothetical maximum time between switches $\bar{t} = \max(t_{2k+1}-t_{2k})$ can be defined as an upper bound for $t$. It will be shown in the following that this bound can be exceeded in case of an attack, and therefore eq. \ref{eq:threshold} is a valid threshold for attack detection. 
\begin{equation}\label{eq:threshold}
\bar{\nu}_{i,fil,OSA}(t_{2k}) = e^{-K\bar{t}} \nu_{i,fil}(t_{2k}) + (1-e^{-K\bar{t}})\bar{\nu} \, ,
\end{equation}
$\bar{t}$ corresponds to the longest time for which $\epsilon_{y,i}= \epsilon_{1,i} -\zeta_{1,i}$ can stay positive. This is the case when $\epsilon_{1,i}$ decreases from its maximum value, $\bar{\epsilon}_1$, to its minimum value, $-\bar{\epsilon}_1$, with a minimum rate $\underline{\dot{\epsilon}}_{1}=\min(\left|\dot{\epsilon}_{1,i}\right|)$. Note that, for this to happen, $\zeta_{1,i}<\epsilon_{1,i}$ during the whole time. This is visualised in Figure~\ref{fig:thresholds}b and results in the following expression for $\bar{t}$
\begin{equation}
\bar{t} = \frac{2\bar{\epsilon}_1}{\underline{\dot{\epsilon}}_1}
\label{eq:max_dwelltime}
\end{equation}
The bounds, $\bar{\epsilon}_1$, $\underline{\dot{\epsilon}}_1$, and $\bar{\zeta}$ are derived in theorem 1, Appendices \ref{app:bound_eps1_dot} and \ref{app:bound_noise} respectively, and shown in eqs. \eqref{eq:bound_eps1_zeta}-\eqref{eq:bound_eps2}.
\begin{equation}\label{eq:bound_eps1_zeta}
    \bar{\epsilon}_1 = \bar{\zeta} = 
    \left[\begin{matrix}
    \bar{\eta}_{i,(1)} + h \bar{\xi}_{i,(2)}\\
    \bar{\eta}_{i,(2)} + h \bar{\xi}_{i,(3)}
    \end{matrix}\right]
\end{equation}
\begin{equation}\label{eq:bound_eps1dot}
\underline{\dot{\epsilon}}_1 = -\left|A_{12} \bar{\epsilon}_{2,i} \right| + M_i
\end{equation}
\begin{equation}\label{eq:bound_eps2}
\bar{\epsilon}_{2,i} = \epsilon_{2,i,0} e^{A_{22}t} - \frac{2 A_{21}\bar{\zeta}-b\Delta u_{i-1}}{A_{22}}
\end{equation}
One can see in eq. \ref{eq:bound_eps2} that $\bar{\epsilon}_{2,i}$ depends on the attack. The threshold is designed assuming no attack, so $\Delta u_{i-1}=0$. Therefore, it is easy to check that if there is an attack, $\epsilon_{2,i}$ can become bigger than $\bar{\epsilon}_{2,i}$ (with $\Delta u_{i-1}=0$). Therefore eq. \ref{eq:threshold} is a valid threshold for attack detection. 

At $t_{2(k+1)}$ this threshold needs to be recalculated using a new initial value of $\nu_{i,fil}(t_{2(k+1)})$, as illustrated in Figure~\ref{fig:thresholds}. This re-initialisation on the signal the threshold is attempting to bound leads to inconsistent detection. Even though an attack can cause detection between recalculations, it is also dependent on the noise behaviour. As before, $\zeta_{1,i}<\epsilon_{1,i}$ needs to hold during $\bar{t}$ for the threshold to be reached, and even though this chance is non-zero in case of an attack, in every period $[t_{2k} \; t_{2k+1}]$ there is a large chance an attack is not detected. Therefore in the next section a threshold is designed that is not dependent on $\nu_{i,fil}$.
\subsection{Multiple-Switches-Ahead (MSA) Threshold}\label{sub:multiple_switch_ahead_thresh}
\noindent The MSA threshold is based on the possible behaviour of $\nu_{fil}$ over more than one switch ahead in time, after a hypothetical occurrence of the worst case behaviour considered for the OSA threshold.
\vspace{-0.6cm}
\begin{figure}[H]
	\centering
	\includegraphics[width=0.9\columnwidth,height=0.7\columnwidth,trim={0.2cm 0.1cm 0.1cm 0cm},clip]{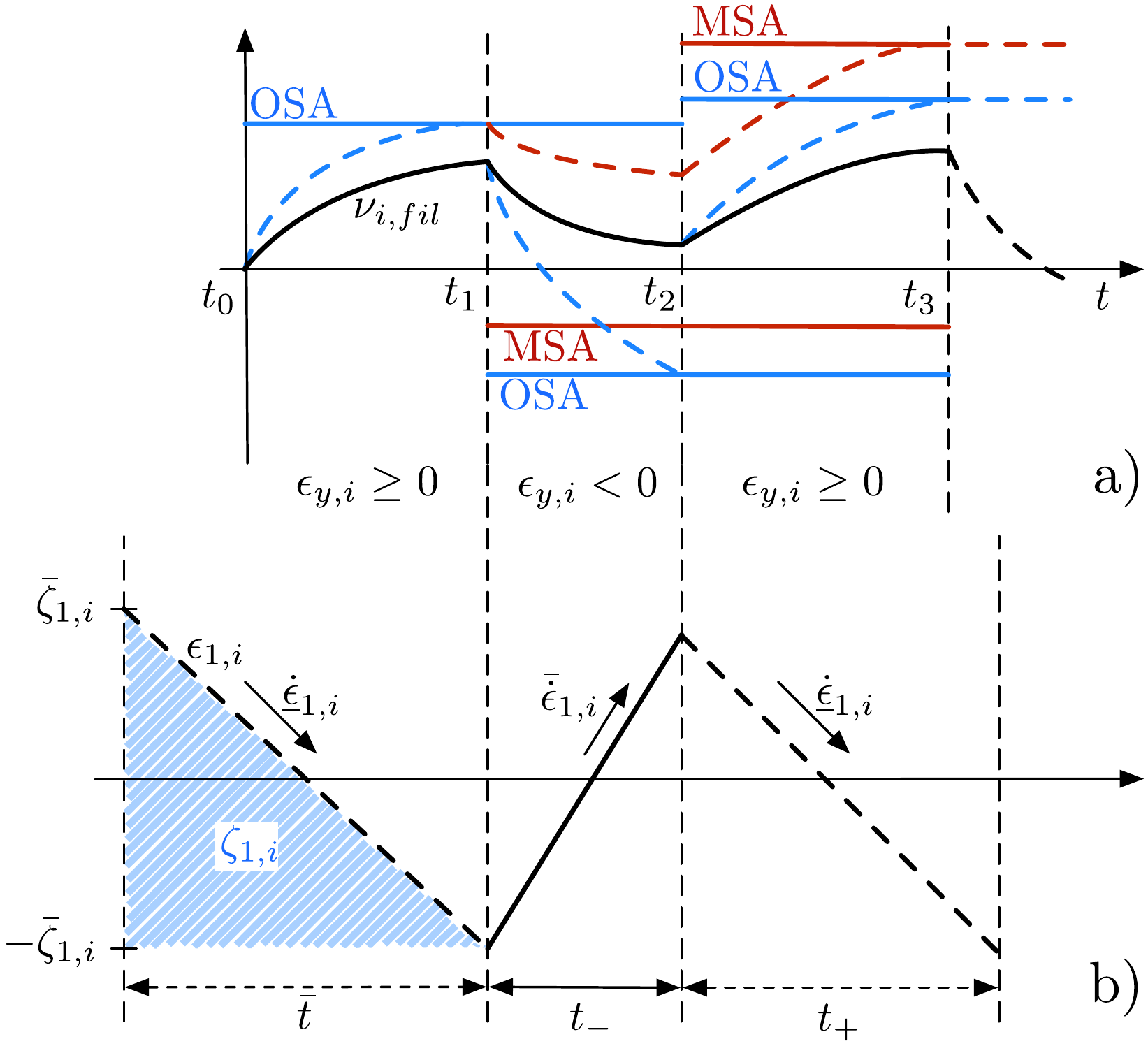}
	\caption{A graphical illustration of the OSA and MSA threshold computations for the first three switching periods of the SMO. a) The EOI, the OSA and the MSA thresholds are drawn with solid lines, while the terms in eqs. \eqref{eq:nu_filter_resp} and \eqref{eq:nu_bar_more} is drawn with dashed lines b) $\bar t$ and $t_+$ are hypothetical maxima calculated at $t_0$ and $t_2$ respectively, and $t_-$ is a measured time. The behaviour of $\epsilon_{1,i}$ is that of the worst case for each time period. These worst cases are used in calculating the thresholds.}
	\label{fig:thresholds}
\end{figure}
\vspace{-0.3cm}
Right before the hypothetical OSA switch at $t_{2k+1}$, $\epsilon_{1,i}=\zeta_{1,i}=-\bar{\zeta}$. At this moment there is a guaranteed switch, and $\text{sign}(\epsilon_{y,i})=-1$, $\dot{\epsilon}_{1,i}>0$, and $\epsilon_{1,i}$ will increase during a period lasting $t_-=t_{2(k+1)}-t_{2k+1}$. At $t_{2(k+1)}$, there will be another switch, making $\epsilon_{1,i}$ decreasing over a hypothetical period $t_+$. Figure~\ref{fig:thresholds}b shows this behaviour for $k=0$.

The MSA threshold considers the largest value that the upper bound $\bar{\nu}_{i,fil}(t_{2(k+1)})$ could attain, starting from a previously computed upper bound $\bar{\nu}_{i,fil}(t_{2k})$, after a known period $t_-$ and a hypothetical period $t_+$. This is shown in eq. \eqref{eq:nu_bar_more}, where $\underline{\nu} = -\left|(A_{11}+P)(\bar{\epsilon}_{1}+\bar{\zeta})\right| + M_i$.
\begin{equation}\label{eq:nu_bar_more}
\begin{aligned}
\bar{\nu}_{i,fil,MSA}(t_{2(k+1)})&= e^{-Kt_+}(e^{-Kt_-} \bar{\nu}_{i,fil}(t_{2k}) \\
                    & \hspace{-0.2cm}  - (1-e^{-Kt_-})\underline{\nu}) + (1-e^{-Kt_+})\bar{\nu}
\end{aligned}
\end{equation}
One can see from eq. \eqref{eq:nu_bar_more} that $\bar{\nu}_{i,fil,MSA}$ is maximal for a big $t_+$. This is the case if $\dot{\epsilon}_{1,i}=\bar{\dot{\epsilon}}_{1,i}$ during $t_-$ and $\dot{\epsilon}_{1,i}=-\underline{\dot{\epsilon}}_{1,i}$ during $t_+$. As $t_-$ is known, the maximum value of $\epsilon_{1,i}$ that can be reached in this time can be calculated as  $\epsilon_{1,i}(t_{2k+1}+t_-)=-\bar{\zeta}+\bar{\dot{\epsilon}}_{1,i}t_-$. The maximum $t_+$ can then be expressed in terms of $t_-$ as in eq. \eqref{eq:eps_ratio}. Here $\bar{\dot{\epsilon}}_{1,i}$ is defined in Appendix \ref{app:bound_eps1_dot} and stated in eq. \eqref{eq:epsdot_bound}.
\begin{eqnarray}
\label{eq:eps_ratio}
t_+ = \frac{\epsilon_{1,i}(t_{2k+1}+t_-)--\bar{\zeta}}{\underline{\dot{\epsilon}}_{1,i}} = \frac{\bar{\dot{\epsilon}}_{1,i}}{\underline{\dot{\epsilon}}_{1,i}} t_-\\
\label{eq:epsdot_bound}
\bar{\dot{\epsilon}}_1 = \left|P (\bar{\zeta} + \bar{\epsilon}_1)\right| + \left|A_{12} \bar{\epsilon}_{2,i}\right| + \left|A_{11}\bar{\zeta}\right| + M_i
\end{eqnarray}
Finally, the threshold that is used to detect an attack is defined as in eq. \eqref{eq:final_threshold}. As both thresholds are guaranteed to have no false attack detection, the combined threshold will also guarantee this. Furthermore by taking the minimum of both thresholds, the threshold is made less conservative:
\begin{equation}
\bar{\nu}_{i,fil}(t_{2k})=\min(\bar{\nu}_{i,fil,OSA}(t_{2k}),\bar{\nu}_{i,fil,MSA}(t_{2k})).
\label{eq:final_threshold}
\end{equation}
\begin{theorem}
$\bar{\nu}_{i,fil}$ only depends on $\nu_{i,fil}$ when it will result in a lower threshold then if it doesn't depend on $\nu_{i,fil}$.
\end{theorem}
\begin{proof}
If $\bar{\nu}_{i,fil}(t_{2k}) = \bar{\nu}_{i,fil,MSA}(t_{2k})$ the threshold will not be dependent on $\nu_{i,fil}$. At every $t_{2(k+1)}$, except $t_2$, a new threshold will be calculated that is only dependent on $\bar{\nu}_{i,fil,MSA}(t_{2k})$. At $t_2$ no previous $\bar{\nu}_{i,fil,MSA}$ is available, so $\bar{\nu}_{i,fil,OSA}(t_0)$ is used. In general $\bar{\nu}_{i,fil,OSA}(t_{2k})$ is dependent on $\nu_{i,fil}(t_{2k})$, however $\bar{\nu}_{i,fil,OSA}(t_0)$ is dependent on $\nu_{i,fil}(t_0)$, which is defined to be 0. This can be done as $\nu_{i,fil}$ is the result of a first order filter that can be initialised.

If $\bar{\nu}_{i,fil}(t_{2k})=\min(\bar{\nu}_{i,fil,OSA}(t_{2k}),\bar{\nu}_{i,fil,MSA}(t_{2k}))$, $\bar{\nu}_{i,fil}(t_{2k})$ will be calculated based on $\bar{\nu}_{i,fil,MSA}(t_{2(k-1)})$ or a lower $\bar{\nu}_{i,fil,OSA}(t_{2(k-1)})$. Furthermore, $\bar{\nu}_{i,fil,OSA}(t_{2k})$ will only become the threshold if it is lower then the $\bar{\nu}_{i,fil,MSA}(t_{2k})$. These statements prove the theorem.
\end{proof}
\subsection{Threshold for Event Triggered Communication}\label{sub:event_triggered_thresh}
\noindent In case of event triggered communication, $\Delta u_{i-1}$ includes both the attack $\phi_i$, and the communication-induced effect $\Delta u_{C,i-1}$ as defined in Section \ref{ssec:attackeffect}. Therefore, without modifications, the observer may falsely detect the communication-induced effect as a cyber attack. The proposed modification to the threshold will prevent this.

The difference between an attack and the event-triggered communication error is that the first will start occurring at communication times $\tau_l$, while the latter will become zero at such times, as an updated value $u_{i-1}$ is received.

Just like the attack, the communication error affects the observer through the dynamics of $\epsilon_{2,i}$, and thus the threshold through $\bar{\epsilon}_{2,i}$ (derived in appendix \ref{app:eps2bound} and stated in eq. \ref{eq:bound_eps2}). This means that for the modified threshold, the increase in $\bar{\epsilon}_{2,i}$ due to $\Delta u_{C,i-1}$ should be taken into account. 

The exact $\Delta u_{C,i-1}$ is not known, however, a worst case scenario exists, assuming there are no local maximums in $u_{i-1}(t)$ between communications. This worst case is when the maximum communication error $\Delta \bar{u}_{C,i-1} \triangleq \tilde{u}_{i-1}(\tau_l)-\tilde{u}_{i-1}(\tau_{l-1})$ occurs constantly since the last communication.

This scenario is implemented by computing all the terms needed for the threshold, using $\bar{\epsilon}_{2,i}$ where $\Delta u_{i-1}=\Delta \bar{u}_{C,i-1}$ for every $t_{2k}$ in the period $[\tau_{l-1} ~ \tau_l]$. These calculations can only be done retroactively, when a communication is received at $\tau_l$. This means that at $\tau_l$ the OSA and MSA thresholds need to be calculated for every $t_{2k}$ in the period $[\tau_{l-1} ~ \tau_l]$.
\section{Attack Estimate} \label{sec:attack_est}
\noindent In this section some preliminary results will be introduced toward the goal of estimating the attack term $\phi$. The method proposed here, is based on \cite{JAHANSHAHI2018212}. This approach is valid only for the case without measurement uncertainty and with continuous observer dynamics. However, in the simulations of section \ref{sec:results} it will be shown that the estimate is also accurate without these assumptions.

First note that without measurement uncertainty $\epsilon_{1,i} = \bar{\epsilon}_1=\bar{\zeta}=[0~0]^\top$, and for continuous observer dynamics this implies $\dot{\epsilon}_{1,i}=\bar{\dot{\epsilon}}_1=[0~0]^\top$. 
Considering the new ideal observer error dynamics, shown in eq. \eqref{eq:obs_err_dyn_ss}, a relation can be found between $\Delta u_{i-1}(t)$ and $\nu_i$. Assuming $\Delta u_{i-1}(t)$ is piecewise constant, the differential equation in the last row of eq. \eqref{eq:obs_err_dyn_ss} can be solved to get eq. \eqref{eq:eps2}.
Substituting this solution in the first row of eq. \eqref{eq:obs_err_dyn_ss}, gives eq. \eqref{eq:nu}.
\begin{equation}
\left[\begin{matrix}
0\\
\dot{\epsilon}_{2,i}(t)
\end{matrix}\right] = \left[\begin{matrix}
A_{11}&A_{12}\\
A_{21}&A_{22}
\end{matrix}\right]\left[\begin{matrix}
0\\
\epsilon_{2,i}(t)
\end{matrix}\right] -
\left[\begin{matrix}
\nu_i(t)\\
b\Delta u_{i-1}(t)
\end{matrix}\right]
\label{eq:obs_err_dyn_ss}
\end{equation}
\begin{equation}
\begin{matrix}
\epsilon_{2,i} \overset{t \to \infty}{=} A^{-1}_{22} b\Delta u_{i-1}(t)
\end{matrix}
\label{eq:eps2}
\end{equation}
\begin{equation}
\nu_i(t) = A_{12}\epsilon_{2,i}(t) = A_{12}A^{ -1}_{22}b\Delta u_{i-1}(t)
\label{eq:nu}
\end{equation}
From this, the estimate for $\Delta u_{i-1}(t)$ can be expressed as in eq. \eqref{eq:dui1}. Here $^+$ indicates the pseudo inverse. Furthermore, as $\nu_i$ is a discontinuous switching term, the EOI $\nu_{i,fil}$ will be used to estimate $\Delta u_{i-1}$ \cite{Edwards2000-vs}.
\begin{equation}
\Delta \hat{u}_{i-1}(t) = b^{-1} A_{22} A^+_{12}\nu_{i,fil}(t)
\label{eq:dui1}
\end{equation}
\section{Simulation Result}
\label{sec:results}
\noindent A CACC-controlled platoon of three vehicles using event triggered communication, equipped with the sliding mode observer presented in this paper, is implemented in Matlab/Simulink. The parameters used in the simulation are shown in tables \ref{tab:parameters1} and \ref{tab:parameters2}. Here the uncertainties are implemented as zero--mean Gaussian white noise, with standard deviations $\sigma_{\eta_i}$ and $\sigma_{\xi_i}$. The simulation scenario considered is shown in Figure \ref{fig:SimScenario}. Results are shown for the same scenario with continuous communication in Figure~\ref{fig:CC_AttackEst} and with event triggered communication in Figure~\ref{fig:ET_AttackEst}.
\vspace{-0.3cm}
\begin{table}[!h]
	\hspace{-0.25cm}\parbox{.45\linewidth}{
		\begin{tabular}{|p{1.1cm}|p{0.9cm}|p{1.25cm}|}
			\hline
			\multicolumn{2}{|c|}{Variable} & Value [unit]\\ \hline
			\vspace{-0.15cm}\multirow{1}{*}{Car} & $\tau$ & $0.1~[-]$\\ \hline
			\multirow{2}{*}{Noise}& $\sigma_{\eta_i}$ & $0.05~[-]$\\ \cline{2-3}
			& $\sigma_{\xi_i}$ & $0.05~[-]$\\ \hline
			\multirow{4}{*}{Network}     & $T_L$ & $0.1~[s]$\\ \cline{2-3}
			& $T_H$ & $1~[s]$\\ \cline{2-3}
			& $\Delta y_{L,(1)}$ & $4~[m]$\\ \cline{2-3}
			& $\Delta y_{L,(2)}$ & $0.5~[m/s]$\\ \hline
			\multicolumn{2}{|c|}{Sim. frequency}& $1000~[Hz]$\\ \hline
		\end{tabular}
		\caption{Simulation Parameters}
		\label{tab:parameters1}
	}
	\hspace{0.5cm}
	\parbox{.45\linewidth}{
		\begin{tabular}{|p{1.4cm}|p{0.6cm}|p{1.25cm}|}
			\hline
			\multicolumn{2}{|c|}{Variable} & Value [unit]\\ \hline
			\multirow{4}{*}{CACC}     & $k_p$ & $0.2~[-]$\\ \cline{2-3}
			& $k_d$ & $0.7~[-]$\\ \cline{2-3}
			& $h$ & $0.7~[s]$\\ \cline{2-3}
			& $r$ & $1.5~[m]$\\ \hline
			\multirow{4}{*}{Observer}& $P$ & $0_{2\times 2}~[-]$\\ \cline{2-3}
			& $M_i$ & $20~[-]$\\ \cline{2-3}
			& $K$ & $2~[-]$\\ \hline
			\multirow{4}{*}{Threshold}& $\epsilon_{2,i,0}$ & $10~[m/s^2]$\\ \cline{2-3}
			& $\bar{\eta}_i$ & $2\sigma_{\eta_i}~[-]$\\ \cline{2-3}
			& $\bar{\xi}_i$ & $2\sigma_{\xi_i}~[-]$\\ \hline
		\end{tabular}
		\caption{Design Parameters}
		\label{tab:parameters2}
	}
\end{table}
\vspace{-1cm}
\begin{figure}[H]
	\centering
	\includegraphics[width=0.9\columnwidth,height = 0.38\columnwidth, trim={0cm 0cm 7cm 20.5cm},clip]{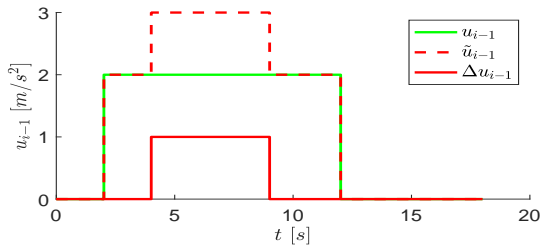}
	\vspace{-0.25cm}
	\caption{Simulation Scenario for attack on communication between car 1 and 2.}
	\label{fig:SimScenario}
\end{figure}
\vspace{-0.2cm}
\noindent The first thing to be noticed in both cases, is the absence of false alarms. Note that the event-triggered threshold in figure~\ref{fig:ET_AttackEst} is only valid at triggering times, indicated with a $*$ marker.

The detection delays in these scenarios are $0.23~[s]$ and $0.6~[s]$, for the Continuous and Event triggered communication respectively. This detection time is scenario specific and depends on many parameters, including the attack and noise magnitudes, and the observer design parameters.

In figure~\ref{fig:ET_AttackEst} two peaks can be seen around $2$ and $12 [s]$. These peaks are caused by the delay in the event triggerred communication. An acceleration is initiated by vehicle 1 at $2.01~[s]$, while the first communication to vehicle 2 is at $2.4~[s]$. During this time there is a nonzero $\Delta u_{i-1,C}$, which the observer will start to estimate.

The attack is introduced at $4.01~[s]$, also asynchronous with the communication, this however has no effect on the observer, as car 2 and the observer will only be affected by to the attack after it has received a communication.

Lastly, note that the threshold converges to a steady state value around $\pm 0.35$, which means that for this scenario all attacks bigger than this will be detected. 
\vspace{-0.3cm}
\begin{figure}[H]
	\centering
	\includegraphics[width=0.9\columnwidth,height = 0.48\columnwidth, trim={0cm 0cm 7cm 18cm},clip]{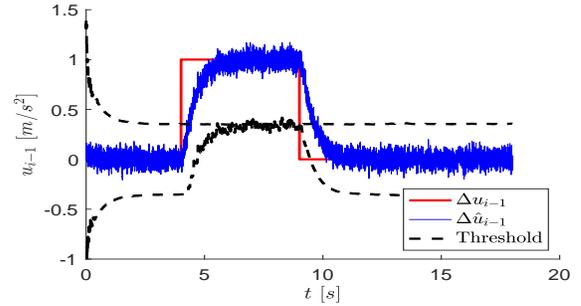}
	\vspace{-0.25 cm}
	\caption{Attack estimation by car 2, continuous communication}
	\label{fig:CC_AttackEst}
\end{figure}
\vspace{-0.6cm}
\begin{figure}[H]
	\centering
	\includegraphics[width=0.9\columnwidth,height = 0.48\columnwidth, trim={0cm 0cm 7cm 18cm},clip]{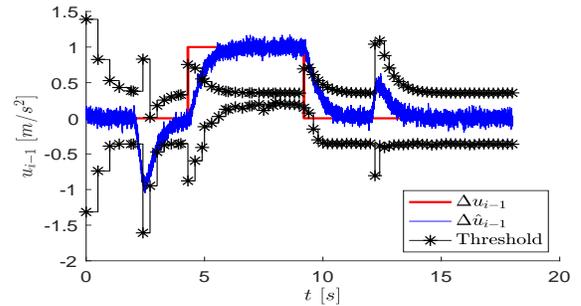}
	\vspace{-0.25 cm}
	\caption{Attack estimation by car 2, event-triggered communication}
	\label{fig:ET_AttackEst}
\end{figure}
\vspace{-0.3cm}
\section{Concluding remarks}%
\label{sec:conclusion}
\noindent In this paper a cyber attack detection and estimation algorithm is presented for a platoon of vehicles using a {\em Cooperative Adaptive Cruise Control} (CACC) algorithm and a realistic, event-triggered {\em Vehicle to Vehicle} communication protocol based on the ETSI ITS G5 standard. A {\em man-in-the-middle} error injection attack is considered on the transmitted intended acceleration of the preceding vehicle, $\Delta u_{i-1}$.

A detection and estimation approach was proposed, based on the so-called {\em Equivalent Output Injection} signal of a {\em Sliding Mode Observer} (SMO). This is combined with an adaptive threshold that is robust against false detection.

The main contribution of this paper is the design of a robust attack detection threshold which incorporates the effects of sensor noise and communication errors. This is done by combining the {\em One-Switch-Ahead} and the {\em Multiple-Switches-Ahead} thresholds. A second theoretical result was provided regarding the stability of the SMO under measurement uncertainties and event-triggered communication. Finally, a preliminary result is proposed that allows to estimate the amplitude of the cyber-attack, under ideal conditions. Simulation results verified the expected behaviour and robustness of the proposed solution, and showed that attack estimation could be attained in practice also under non-ideal conditions.

In future work, we would like to derive theoretical results on the attack estimation, which are also valid in non-ideal conditions, and extend the approach to the case of more general (non-)linear dynamical systems.
\appendix
\subsection{Upper and Lower bound for $\dot{\epsilon}_{1,i}$}
\label{app:bound_eps1_dot}
\noindent $\bar{\dot{\epsilon}}_{1,i}$ is defined as $\max(|\dot{\epsilon}_{1,i}|)$ and $\underline{\dot{\epsilon}}_{1,i}=\min(|\dot{\epsilon}_{1,i}|)$, they can be constructed by substituting the upper and lower bounds of all terms, as presented in eqs. \eqref{eq:low_bound1}-\eqref{eq:low_bound3}, into eq. \eqref{eq:obs_err_dyn_12}. This gives eqs. \eqref{eq:min_epsdot} and \eqref{eq:max_epsdot} as expressions for the bounds. Only the upper bound on $A_{11}\zeta_{1,i}$ in eq. \eqref{eq:low_bound3} is non-trivial and will be proved in theorem \ref{theorem:dot_eps_eff}.
\begin{eqnarray}
-|P(\bar{\epsilon}_{1}+\bar{\zeta})|\leq-P(\epsilon_{1,i}-\zeta_{1,i})\leq0
\label{eq:low_bound1}\\
-\left|A_{12} \bar{\epsilon}_{2,i} \right|\leq A_{12} \epsilon_{2,i} \leq\left|A_{12} \bar{\epsilon}_{2,i} \right|
\label{eq:low_bound2}\\
-\left|A_{11}\bar{\zeta}\right|\leq A_{11}\zeta_{1,i} \leq0
\label{eq:low_bound3}
\end{eqnarray}
\begin{equation}
\underline{\dot{\epsilon}}_{1,i} = -\left|A_{12} \bar{\epsilon}_{2,i} \right| + M_i
\label{eq:min_epsdot}
\end{equation}
\begin{equation}
\bar{\dot{\epsilon}}_{1,i} = |P(\bar{\epsilon}_{1,i}+\bar{\zeta})|+\left|A_{12} \bar{\epsilon}_{2,i} \right|+\left|A_{11}\bar{\zeta} \right| + M_i
    \label{eq:max_epsdot}
\end{equation}
\begin{theorem}
\label{theorem:dot_eps_eff}
Averaged over a maximum dwell time scenario (equations \eqref{eq:max_dwelltime} and \eqref{eq:eps_ratio}), the effect of $A_{11}\zeta_{1,i}$ on $\dot{\epsilon}_{1,i}$ can only be negative.
\end{theorem}
\begin{proof}
In this scenario, for $\epsilon_{1,i}-\zeta_{1,i} >0$, $\epsilon_{1,i}$ is monotonically decreasing from an upper bound lesser or equal to $\bar{\zeta}$ to the lower bound $-\bar{\zeta}$. As $A_{11}$ is a positive definite matrix, the effect of $A_{11}\epsilon_{1,i}$ on $\dot{\epsilon}_{1,i}$, averaged over the maximum dwell time scenario, will be non-positive. As $\zeta_{1,i} <\epsilon_{1,i}$, the effect of $A_{11}\zeta_{1,i}$ on $\dot{\epsilon}_{1,i}$ will be negative.
\end{proof}
\subsection{Upper bound for $\epsilon_{2,i}$}
\label{app:eps2bound}
\noindent By taking the second row of eq. \eqref{eq:obs_error_dyn},  and bounding $A_{21}(\epsilon_{1,i}-\zeta_{1,i})<\left|A_{21}(\bar{\epsilon}_{1}+\bar{\zeta})\right|$ we obtain the differential inequality in eq. \eqref{eq:eps2_diffineq}.
\begin{equation}
    \dot{\epsilon}_{2,i} < \left|A_{21}(\bar{\epsilon}_{1}+\bar{\zeta})\right| + A_{22} \epsilon_{2,i} - b\Delta u_{i-1}
    \label{eq:eps2_diffineq}
\end{equation}
Using lemma 1.1.1 in \cite{Lakshmikantham2015-il}, this gives the expression for $\bar{\epsilon}_{2,i}$ in eq. \eqref{eq:max_eps_3}. Here $\Delta u_{i-1}$ is assumed to be piecewise constant.
\begin{equation}
\epsilon_{2,i} \leq \bar{\epsilon}_{2,i} = \epsilon_{2,i,0} e^{A_{22}t} - \frac{\left|A_{21}(\bar{\epsilon}_{1,i}+\bar{\zeta})\right|-b\Delta u_{i-1}}{A_{22}}
\label{eq:max_eps_3}
\end{equation}
\subsection{Upper bound for $\zeta_{1,i}$}
\label{app:bound_noise}
\noindent Looking at the definition of $\zeta_{1,i}$ in equations \eqref{eq:error_dynamics_definitions} and \eqref{eq:reduced_linsys}, its bound can be expressed in terms of the bounds on the individual noise terms (defined in assumption \ref{ass:measurement_uncertainties}) as
\begin{equation}
\bar{\zeta}=\left[\begin{matrix}
    \bar{\eta}_{i,(1)} + h \bar{\xi}_{i,(2)}\\
    \bar{\eta}_{i,(2)} + h \bar{\xi}_{i,(3)}
    \end{matrix}\right]
\label{eq:zeta_bar}
\end{equation}
\begin{spacing}{0.87}
\bibliographystyle{IEEEtran}
\bibliography{references}
\end{spacing}

\end{document}